\documentclass[final,8pt,5p,twocolumn]{elsarticle}

\usepackage{amsmath,amssymb,amsfonts}

\usepackage{ulem}

\usepackage{xcolor}

\usepackage{tabularx,lipsum,graphicx,tabularray}  

\usepackage[hyphens]{url}
\usepackage[colorlinks=true, urlcolor=blue, pdfborder={0 0 0}]{hyperref}
\hypersetup{
  colorlinks   = true, 
  urlcolor     = blue, 
  linkcolor    = blue, 
  citecolor   = red 
}
\hypersetup{breaklinks=true}
\usepackage{subcaption}
\usepackage{booktabs}
\usepackage{multirow}
\usepackage{tabularx}	
\usepackage{color,colortbl}

\usepackage{algorithmic}
\usepackage{algorithm}

\usepackage{graphicx}

\usepackage{amsthm} 

\biboptions{sort&compress}

\def\XX{\mathbb{X}}

\newtheorem{theorem}{Theorem}
\newtheorem{statement}{Statement}

\def\XX{\mathbb{X}}

\makeatletter
\newcommand{\ion}[2]{%
  #1$\;$%
  \if b\expandafter\@car\f@series\relax\@nil
    \begingroup 
      \sbox0{\rmfamily\mdseries\textsc{v}}%
      \resizebox{!}{\ht0}{\rmfamily\@Roman{#2}}%
    \endgroup
  \else
    \textsc{\rmfamily\@roman{#2}}%
  \fi
}
\makeatother

\makeatletter
\def\@author#1{\g@addto@macro\elsauthors{\normalsize%
    \def\baselinestretch{1}%
    \upshape\authorsep#1\unskip\textsuperscript{%
      \ifx\@fnmark\@empty\else\unskip\sep\@fnmark\let\sep=,\fi
      \ifx\@corref\@empty\else\unskip\sep\@corref\let\sep=,\fi
      }%
    \def\authorsep{\unskip,\space}%
    \global\let\@fnmark\@empty
    \global\let\@corref\@empty  
    \global\let\sep\@empty}%
    \@eadauthor={#1}
}
\makeatother

\journal{Blockchain: Research and Applications}

\begin{document}

\begin{frontmatter}



\title{DIT: Dimension Reduction View on Optimal NFT Rarity Meters}

\author[MIPT]{Dmitry Belousov}
\ead{belousov.da@phystech.edu}
\affiliation[MIPT]{organization={Moscow Institute of Physics and Technology},
            city={Moscow},
            country={Russia}}

\author[Sk,HSE]{Yury Yanovich}
\ead{{Corresponding author*}{y.yanovich@skoltech.ru}}
\affiliation[Sk]{organization={Skolkovo Institute of Science and Technology},
            city={Moscow},
            country={Russia}}

\affiliation[HSE]{organization={Faculty of Computer Science, HSE University},
            country={Russia}}

\begin{abstract}
    Non-fungible tokens (NFTs) have become a significant digital asset class, each uniquely representing virtual entities such as artworks. These tokens are stored in collections within smart contracts and are actively traded across platforms on Ethereum, Bitcoin, and Solana blockchains. The value of NFTs is closely tied to their distinctive characteristics that define rarity, leading to a growing interest in quantifying rarity within both industry and academia. While there are existing rarity meters for assessing NFT rarity, comparing them can be challenging without direct access to the underlying collection data. The Rating over all Rarities (ROAR) benchmark addresses this challenge by providing a standardized framework for evaluating NFT rarity. This paper explores a dimension reduction approach to rarity design, introducing new performance measures and meters, and evaluates them using the ROAR benchmark. Our contributions to the rarity meter design issue include developing an optimal rarity meter design using non-metric weighted multidimensional scaling, introducing Dissimilarity in Trades (DIT) as a performance measure inspired by dimension reduction techniques, and unveiling the non-interpretable rarity meter DIT, which demonstrates superior performance compared to existing methods.
\end{abstract}


\end{frontmatter}


\section{Introduction}
\label{sec:intro}

Non-fungible tokens (NFTs) have become a transformative digital asset, with each token representing a unique identifier often linked to digital objects such as artworks~\cite{Oliveira2018,Angelo2020,Wang2021}. Some NFTs are stored in collections within blockchain-based smart contracts and are actively traded on platforms such as OpenSea, Solanart, and Singular across blockchains including Ethereum, Bitcoin, and Solana \cite{Wood2016,Nakamoto2008,Yakovenko2018}. 

The value of NFTs is determined by a set of unique traits that define each token within a collection (see Figure \ref{fig:BAYC173}). While all NFTs are inherently unique, their trade values vary significantly based on factors such as rarity. This has spurred interest in measuring and quantifying NFT rarity, both within the industry and in academic research, particularly in the context of blockchain-based information systems.

Rarity meters have been developed to assess NFT rarity~\cite{Rarity.tools2021}. However, comparing these meters is challenging due to the lack of convenient access to underlying NFT collection data~\cite{Krasnoselskii2023}. To address this, the Rating over all Rarities (ROAR) benchmark was introduced~\cite{Belousov2024}, providing a standardized framework for evaluating NFT rarity. This benchmark is particularly relevant for blockchain-based systems, as it enhances trust and transparency in NFT valuation. Nevertheless, the performance measure in ROAR makes it computationally intensive and slow to compute.

Rarity serves as a one-dimensional representation of the dataset, where tokens with higher values are deemed more valuable. This paper explores a dimension reduction perspective on rarity design, resulting in new performance measures and meters. Our work contributes to the broader field of blockchain-based information systems by addressing the challenges of rarity quantification in a decentralized and trustless environment.

The main contributions of our paper are as follows:
\begin{enumerate}
    \item Formulate an optimal rarity meter design using non-metric weighted unidimensional scaling.
    \item Generalize a computational method for non-metric unidimensional scaling for measurements weighted with arbitrary weights.
    \item Introduce Dissimilarity in Trades (DIT) as a performance measure, inspired by dimension reduction and computationally simpler than traditional weighted correlation methods.
    \item Introduce the non-interpretable rarity meter DIT, which demonstrates superior performance compared to existing competitors at the cost of losing interpretability.
\end{enumerate}

\begin{figure}[h!]
    \includegraphics[width=0.95\linewidth]{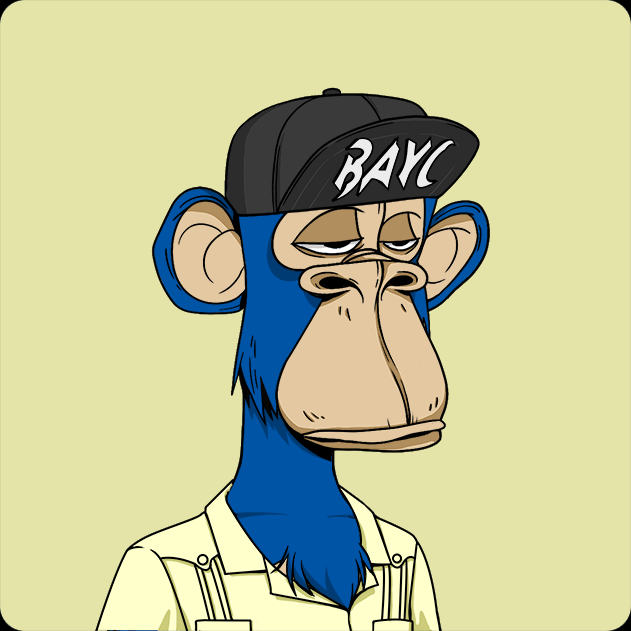}
    \centering
    \caption{NFT $n=173$ is part of the Bored Ape Yacht Club (BAYC) collection, which consists of $10,000$ unique tokens. The contract address for this collection is \texttt{0xbc4ca0eda7647a8ab7c2061c2e118a18a936f13d}. BAYC tokens have seven distinct traits: background, clothes, eyes, hats, earrings, fur, and mouth. For NFT $n=173$, the trait values are as follows: Yellow, Guayabera, Bored, Bayc Flipped Brim, None, Blue, and Bored.}
    \label{fig:BAYC173}
\end{figure}

The remainder of the paper is structured as follows. Section~\ref{sec:RelatedWork} reviews related work on NFT rarity, rarity meter designs, and dimension reduction techniques. Section~\ref{section:Background} provides formal definitions and an overview of state-of-the-art rarity meters. Section~\ref{sec:DimensionReduction} introduces the dimension reduction framework for designing optimal rarity meters. Sections~\ref{sec:OptimalRarityMeterDesign} and~\ref{section:DITMethodology} present the theoretical foundations for the DIT rarity meter and its methodology, respectively. Section~\ref{sec:Experiments} details the numerical experiments and performance evaluations conducted using the ROAR benchmark. Section~\ref{sec:Discussion} discusses the implications of our findings for blockchain-based information systems. Finally, Section~\ref{sec:Conclusion} concludes the paper and outlines future research directions.

\section{Related Work}
\label{sec:RelatedWork}

\subsection{The Demand on Rarity}

Rarity has long been a driving force in various markets, from collectibles to luxury goods, as it appeals to the human desire for uniqueness and exclusivity. Collectors often value rare items not just for their utility but for their scarcity, which can significantly influence market prices. For instance, in the numismatic market, rare coins often command higher prices due to their limited mintage and historical significance \cite{PauleVianez2022}. Similarly, brands leverage rarity to create a sense of exclusivity, thereby increasing consumer demand and profitability \cite{Chailan2018}. However, quantifying the exact value of rarity remains a complex task, as it often intertwines with other factors such as design, historical context, and market trends.

The impact of rarity on pricing has been empirically studied in various contexts. For example, \cite{Koford1998} examined the effect of original mintage on the prices of Morgan Dollars and Liberty Head Eagles from 1878 to 1921. By analyzing 95 mintages, the authors isolated the influence of rarity from other factors such as design and beauty. Their linear model, which used the logarithms of coin age and mintage as features, demonstrated that rarity had a stronger influence on market prices than age, particularly in trade data from 1988 and 1993 \cite{Hastie2009}.

In the realm of collectible card games, \cite{Hughes2022} investigated the impact of rarity on the secondary market prices of ``Magic: the Gathering'' (MTG) cards. While MTG cards have intrinsic utility in gameplay, their rarity--ranging from common to mythic rare--also plays a significant role in determining their market value. Hughes collected price data from TCGplayer for sixteen recent releases over a six-week period in March-April 2019. His linear model, which incorporated in-game attributes, rarity indicators, and the odds of obtaining the card, showed a moderate but statistically significant fit to the trade data. The study found that publisher-assigned rarity indicators and the odds of obtaining the card were important factors influencing card prices in the secondary market.

Beyond collectibles, rarity also plays a crucial role in other domains such as luxury goods, art, and even conservation. For example, the rarity of certain species can drive eco-tourism, while the scarcity of luxury items can enhance their desirability and market value. These examples underscore the pervasive influence of rarity across different markets and its importance in understanding consumer behavior and pricing dynamics.

\subsection{Rarity Meter Designs}

The design of rarity meters for NFT collections has evolved significantly since the 2021 surge in NFT popularity. Unlike traditional collectibles such as coins or MTG cards, NFTs are characterized by structured sets of traits, which necessitate specialized approaches to quantify rarity. This subsection reviews the development of rarity meters, focusing on their design principles and methodologies.

The first widely adopted rarity meter was introduced by Rarity.tools \cite{Rarity.tools2021}, which proposed a formula based on the summation of the inverse frequencies of trait values. This approach, while simple, became the de facto industrial standard and was integrated into various platforms~\cite{SlashdotMedia2022}. However, the formula lacks theoretical justification and does not account for the specific characteristics of individual collections, such as trait dependencies or visual aesthetics.

In response to the limitations of Rarity.tools, academic efforts have sought to develop more robust and interpretable rarity meters. The KRAMER project, introduced for the Kanaria NFT collection hackathon in 2021~\cite{Krasnoselskii2022}, proposed a workflow to compare rarity meters using trade data, a machine learning pipeline to train an optimal rarity meter, and a tournament score as building block for rarity meter. The KRAMER meter, detailed in \cite{Krasnoselskii2023}, computes a weighted correlation between pairwise relative rarities and sell prices for recent deals, providing a more data-driven approach to rarity assessment. However, despite its academic rigor, KRAMER has seen limited adoption in the industry.

OpenSea, the largest NFT marketplace by trade volume~\cite{White2022}, introduced OpenRarity in collaboration with Curio, icy.tools, and PROOF~\cite{OpenRarity2022}. OpenRarity focuses on traits entropy and incorporates unique items and the number of non-default traits to align with expert expectations. This community-driven formula has been adopted by OpenSea and represents a significant step forward in rarity meter design.

In addition, a new approach called NFTVis \cite{Yan2023} has been suggested to assess the performance of NFTs based on visual features. This method estimates a token's rarity separately based on its traits and associated image. The NFTGo \cite{NFTGo2023} is used as a traits-based rarity score. NFTGo is calculated using the Jaccard distance, while the image-based score is obtained as the average difference between the considered image and others in the collection. This dual approach provides a more comprehensive assessment of NFT rarity, combining both trait-based and visual feature-based metrics.

The ROAR benchmark \cite{Belousov2024} represents a significant advancement in the field by providing a comprehensive dataset of 100 popular NFT collections and evaluating the performance of various rarity meters. The study introduced the ROAR rarity meter, a combination of Rarity.tools, KRAMER, OpenSea, and NFTGo, which demonstrated superior performance compared to other meters. This benchmark serves as a valuable resource for testing and refining new rarity meter designs.

While the aforementioned approaches have contributed to the understanding of NFT rarity, they are not without limitations. For instance, the study \cite{Lee2024} provides valuable insights into the relationship between rarity and NFT prices using formal concept analysis (FCA). The authors demonstrate that rarity significantly influences NFT prices, particularly within the medium rarity range, while other factors such as uniqueness, appeal, and image naturalness become more relevant at extreme rarity levels. However, the study does not provide a formalized algorithm for rarity assessment, and its analysis is limited to a single BAYC collection.

\subsection{Rarity Meters Applications}  
\label{subsec:RarityMetersApplications}  

Rarity meters have found diverse applications in the NFT ecosystem, ranging from pricing mechanisms to investment analysis and recommendation systems. These applications highlight the importance of rarity as a key factor influencing NFT valuation and market dynamics.  

One of the earliest studies on NFT market structure by~\cite{Nadini2021} analyzed transactions and metadata from Ethereum and WAX blockchains, revealing that trading history is a strong predictor of NFT prices, while visual properties play a lesser role. Building on this foundation, \cite{Mekacher2022} demonstrated a positive correlation between rarity and sale prices, as well as a negative correlation between rarity and sales volume, using a dataset of Ethereum NFT collections. These findings underscore the significance of rarity in determining NFT market behavior and provide empirical evidence for its role in pricing dynamics.  

From an economic perspective, \cite{Schaar2022} applied hedonic regression to analyze the investment performance of the CryptoPunks NFT collection. The study found that rare traits significantly impact NFT prices, aligning with traditional methods used in art and real estate markets. This approach not only validates the importance of rarity in digital asset pricing but also bridges the gap between traditional and digital markets, offering a robust framework for future research.  


Recent advancements in deep learning have enabled the development of hybrid recommendation systems for NFTs. For example, \cite{Aydogdu2024} proposed a system that combines rarity metrics with visual and transactional data to provide personalized NFT recommendations. This approach leverages transfer learning techniques to predict NFT prices and sales characteristics, offering valuable insights for investors and collectors. Similarly, \cite{Pala2024} utilized transfer learning of visual attributes to predict NFT prices, demonstrating the potential of machine learning in enhancing rarity-based applications. These systems not only improve the accuracy of rarity assessments but also expand the practical utility of rarity meters in real-world scenarios.  

By integrating findings from these studies, it becomes evident that rarity meters have evolved from simple tools for assessing scarcity to sophisticated systems that influence pricing, investment strategies, and user recommendations. This progression underscores the growing importance of rarity in the NFT ecosystem and highlights the need for continued innovation in rarity meter design and application.

\subsection{Unidimensional Scaling}  
\label{subsec:UnidimensionalScaling}  

Rarity meters employ dimensionality reduction to map complex datasets into a one-dimensional space, enhancing analytical simplicity and interpretability. While multidimensional scaling (MDS) is a widely used technique for such tasks, reducing the output dimension to one introduces unique challenges. This special case, termed unidimensional scaling (UDS), imposes constraints inherent to linear ordering: points cannot be rearranged without disrupting spatial relationships or introducing discontinuities. These constraints distinguish UDS from higher-dimensional MDS and demand specialized algorithms tailored to preserve topological integrity in one-dimensional representations.  

The development of dimensionality reduction methods began with linear approaches like principal component analysis (PCA) \cite{Pearson1901}. As datasets grew more complex, non-linear methods such as MDS gained prominence \cite{Cox2008}, with applications extending to non-metric dissimilarities and reliability analysis \cite{Leeuw2009}. The emergence of the manifold hypothesis \cite{Seung2000,Levina2004,Gomtsyan2019} further advanced the field by positing that high-dimensional data often lies near intrinsic low-dimensional manifolds. This insight spurred innovations in manifold learning algorithms \cite{Belkin2003,Huo2008,VectorDiffusionMaps,Bernstein2015b}, which aim to uncover latent geometric structures.  

Neural networks have also contributed to dimensionality reduction through autoencoders \cite{Baldi2012,Bank2023}, which learn compact representations of data by encoding it into lower-dimensional latent spaces. Recent work has extended this framework to model implicit manifolds \cite{Ross2022}, demonstrating their potential for advancing data analysis paradigms.  

In this paper, we demonstrate that relaxing interpretability requirements reduces the problem to a weighted UDS formulation. While prior work by Pliner \cite{Pliner1984,Pliner1986,Pliner1996} explores distance-based weighting using power transformations, our approach generalizes the weight structure. Although Pliner’s framework acknowledges weighted configurations, it lacks a universal proof of optimality and a concrete optimization algorithm--gaps we address here. Additionally, insights from \cite{Ge2005} and \cite{Hubert2002} on UDS in constrained settings, such as predefined object ordering, underscore its utility in simplifying complex datasets. These advances inform our adaptation of UDS to NFT rarity meters, enabling robust quantification of rarity in one-dimensional spaces.

\section{Background: Interpretable Rarity Meters}
\label{section:Background}

Rarity meters are essential tools for evaluating the uniqueness and value of NFTs within a collection. These meters assign a rarity score to each token based on its traits, with higher scores indicating greater rarity. The design of interpretable rarity meters--analytical functions of traits--is crucial for ensuring transparency and consistency in the evaluation process. This section provides an overview of the key concepts and methodologies used in the development of interpretable rarity meters, drawing from established approaches and performance evaluation techniques.

\subsection{Formal Definitions and Performance Measures}
\label{subsection:PM}

An NFT collection with \( N \) tokens and \( T \) traits is denoted as \( \mathbb{X}_N = \{X_n\}_{n=1}^N \), where \( X_n = (x_{n,1},\dots,x_{n,T}) \). A \textbf{rarity meter} for such a collection is defined as a function \( R: \mathbb{X}_N \rightarrow [0,\infty) \), where a higher value \( R(X) \) indicates greater rarity for token \( X \in \mathbb{X}_N \).

The performance of a rarity meter is evaluated using a \textbf{weighted correlation} \( F_{wc}(R; \mathbb{X}_N) \), which ranges from \([-1, 1]\). Weighted correlation is calculated based on closely timed pairs of NFT deals, providing a measure of how well the rarity meter aligns with market dynamics. A higher value of \( F_{wc}(R; \mathbb{X}_N) \) indicates better performance of the rarity meter \( R \) on the collection \( \mathbb{X}_N \). This measure ensures that the rarity meter aligns well with the market value of the tokens, as it captures the relationship between rarity and price.

\subsection{State-of-the-Art Rarity Meters}

Several state-of-the-art rarity meters have been developed, each employing different methodologies to calculate rarity scores:

\subsubsection{Rarity.tools} This approach computes individual scores for all traits and sums them to provide the resulting rarity. The individual scores are inverse fractions of the trait rarity value in the collection. Additionally, Rarity.tools introduces a meta-trait called \textbf{traits count}, which represents the number of non-\texttt{None} traits of an individual token. The \textbf{Rarity.tools rarity meter} \( R_{rt} \) for a given token \( X_k \in \mathbb{X}_N \) is defined as:
   \[
   R_{rt}(X_k) = \sum_{t=1}^{T+1} \texttt{score\_rt}_t(X_k),
   \]
   where for \( t=1, \dots, T+1 \):
   \[
   \texttt{score\_rt}_t(X_k) = \frac{N}{\#\{x \in \mathbb{X}_{N,t}|x= x_{k,t}\}}.
   \]

\subsubsection{KRAMER} The KRAMER rarity meter computes individual scores for all traits and provides the resulting rarity as their weighted sum, with coefficients that maximize the performance measure \( F_{wc} \). The individual scores are derived from a tournament-style comparison of trait groups, where fewer tokens with a given trait value result in a higher score. The \textbf{KRAMER rarity meter} \( R_{kr} \) for a given token \( X_k \in \mathbb{X}_N \) is defined as:
   \[
   R_{kr}(X_k) = \sum_{t=1}^T \alpha_t \cdot \texttt{score\_kr}_t(X_k),
   \]
   where \( \alpha_1, \dots, \alpha_N \) are the solution of:
   \[
   F_{wc}(R;\mathbb{X}_N) \rightarrow \max_{R \in \mathcal{R}}.
   \]

\subsubsection{OpenRarity} This approach calculates the Shannon information of a token and normalizes it by the collection average. OpenRarity introduces additional heuristics, such as the \textbf{Double Sort} and \textbf{Trait Count}, to refine the rarity calculation. The \textbf{OpenRarity rarity meter} \( R_{or} \) for a given token \( X_k \in \mathbb{X}_N \) is defined as:
   \[
   R_{or}(X_k) = R_{or,0}(X_k) + \texttt{reg}(X_k),
   \]
   where:
   \[
   R_{or,0}(X_k) = \frac{I(X_k)}{\mathbb{E} I(X)}, \quad \texttt{reg}(X_k) = x_{k,0} \cdot \frac{(T + 1) \log N}{\mathbb{E} I(X)}.
   \]

\subsubsection{NFTGo} The NFTGo rarity meter calculates the Jaccard distance in the trait space between a specific NFT and the entire collection. The Jaccard distance is normalized to ensure the rarity score falls within a specific range. The \textbf{NFTGo rarity meter} \( R_{go} \) for a given token \( X_k \in \mathbb{X}_N \) is defined as:
   \[
   R_{go}(X_k) = \text{normalize}\left(\sum_{n=1}^N \left(1 - \texttt{JD}(X_k, X_n)\right)\right),
   \]
   where normalize maps \( R_{go} \) such that the minimum value on the collection is \( 0 \) and the maximum value equals \( 100 \).

\subsubsection{ROAR Rarity Meter} Inspired by the state-of-the-art approaches, the \textbf{ROAR rarity meter} is proposed as an ensemble of other rarity building blocks. The ROAR rarity meter combines elements from Rarity.tools, KRAMER, OpenRarity, and NFTGo to optimize the performance measure \( F_{wc}(R; \mathbb{X}_N) \). The \textbf{ROAR rarity meter} \( R_{rr} \) for a given token \( X_k \in \mathbb{X}_N \) is defined as:
\[
R_{rr}(X_k) = \arg\max_{R \in \mathcal{R}_{rr}} F_{wc}(R;\mathbb{X}_N),
\]
where \( \mathcal{R}_{rr} \) is the set of non-negative combinations of \( \texttt{score\_rt}_t \), \( \texttt{score\_kr}_t \), \( \texttt{score\_or}_t \), and \( R_{go} \).

\subsection{Performance Evaluation}

The dataset used for evaluating the performance of rarity meters includes \textbf{100 NFT collections}, each with traits and trade data \cite{Belousov2024}. To assess rarity meters, the trade data for each collection is split into training (70\% of earliest transactions) and testing (30\% of later transactions) sets, ensuring the model is trained on historical data and evaluated on recent transactions.

The performance of a rarity meter is measured using the \textbf{weighted correlation} \( F_{wc}(R; \mathbb{X}_N) \), where higher values indicate better alignment between rarity scores and market value. To compare rarity meters across collections, \textbf{performance profiles} are used. These profiles show the fraction of collections where the performance measure \( F \) for a given rarity meter is within a certain distance \( \tau \) of the best-performing meter on that collection. The \textbf{performance profile} \( \rho \) of a rarity meter \( R_m \) in a set of rarity meters \( R_1, \dots, R_M \) is defined as:
\[
\rho(\tau; R_m) = \frac{1}{C} \sum_{c=1}^C \left[F_{m,c} + \tau \geq \max\{F_{i,c}\}_{i=1}^M\right],
\]
where \( [A] \) is the indicator of an event \( A \), and \( C = 100 \) is the number of collections in the dataset.

This evaluation methodology ensures a robust and fair comparison of rarity meters, providing insights into their effectiveness across a diverse range of NFT collections.

\section{Dimension Reduction View on Rarity Meters}  
\label{sec:DimensionReduction}  

The concept of rarity in NFTs is inherently tied to market dynamics, as the only objective measure of rarity is derived from trade data. However, NFT prices are highly volatile, influenced by factors such as fluctuations in native blockchain cryptocurrency prices and shifts in market interest \cite{Krasnoselskii2023}. To mitigate this volatility, it is essential to consider pairs of deals that are close in time rather than individual transactions. This approach provides a more stable basis for assessing rarity and its relationship to market value.  

The weighted correlation \( F_{wc} \) suggests that the vectors of pairwise rarities and deal prices should exhibit similarity. Close-in-time deals lead to the construction of an NFT token dissimilarity matrix, where values closer to \( 0 \) indicate greater similarity between two tokens. For the weighted correlation \( F_{wc} \), the dissimilarity of a pair of deals \( \{(t_d, i_d, p_d)\}_{d=1}^2 \) is defined as  
$ 
k(t_1, t_2) \left| \ln \frac{p_1}{p_2} \right|,  
$
where \( k(t_1, t_2) \) is a kernel function that assigns higher weights to deals closer in time. 

By aggregating all deals for a specific pair of NFT indices \( (i, j) \), we obtain the weighted dissimilarity \( \delta_{ij} \) and its corresponding weight \( w_{ij} \):  
\begin{eqnarray*}  
w_{ij} &=& \sum_{d_1, d_2: \;\; (i_{d_1}, i_{d_2}) = (i, j)} k(t_{d_1}, t_{d_2}), \\  
\delta_{ij} &=& \frac{1}{w_{ij}} \sum_{d_1, d_2: \;\; (i_{d_1}, i_{d_2}) = (i, j)} k(t_{d_1}, t_{d_2}) \left| \ln \frac{p_{d_1}}{p_{d_2}} \right|.  
\end{eqnarray*}  

Rarity meters provide a one-dimensional output, with the native dissimilarity for tokens \( (i, j) \) given by:  
\[  
d_{ij}(R) = \left|R(X_i) - R(X_j)\right|.  
\]  
The DIT performance of the rarity meter induced by dissimilarity can be assessed using non-metric weighted unidimensional scaling (NWUDS):  
\[  
F(R; \mathbb{X}_N) = \sqrt{\frac{\sum_{i,j=1}^N w_{ij} \left( d_{ij}(R) - \delta_{ij} \right)^2}{\sum_{i,j=1}^N w_{ij} d_{ij}(R)^2}}.  
\]  
To eliminate the effect of the scale factor, normalization and square root are applied.

The key advantage of DIT's $F$ over ROAR's $F_{wc}$ lies in its computational efficiency. While $F_{wc}$ requires iterating through pairwise deals for each computation of the performance measure, $F$ only needs to process the dissimilarities and weights matrices once. After this initial step, $F$ operates exclusively on $N \times N$ matrices, which is significantly faster than working with the vector of pairwise deals. This efficiency makes $F$ more practical for large-scale applications.

The optimal design for a rarity meter under \( F(R; \mathbb{X}_N) \) involves finding a set of numbers \( \vec{R} = (R_1, \dots, R_N)^T = (R(X_1), \dots, R(X_N))^T \) that minimizes the NWUDS objective function. This problem can be viewed as one-dimensional non-metric weighted multidimensional scaling, where the goal is to preserve the pairwise dissimilarities \( \delta_{ij} \) in the one-dimensional rarity space.  

Given that $F$ is shift-invariant, the solution can be shifted to ensure that \( \min_{n=1, \dots, N} R_n = 1 \), ensuring non-negativity of rarity scores. Additionally, since \( F(\vec{R}; \mathbb{X}_N) = F(-\vec{R}; \mathbb{X}_N) \), the sign of the solution can be chosen to maximize the total rarity-weighted log-price, aligning the rarity scores with market value.

At the same time, $F$ is not scale-invariant. This motivates the adjustment of a scalar factor for the rarity meters, which were originally not designed for $F$:

\begin{statement}
\label{statement:scale}
    The optimal scalar factor $\alpha$ that minimizes the performance measure $F$ in the set of rarity meters $\{\alpha R \mid \alpha > 0\}$ for a fixed rarity meter $R$ is given by:
    $$
    \arg\min_{\alpha > 0} F(\alpha R; \XX_N) = \frac{\sum\limits_{i, j = 1}^N \omega_{ij} \delta^2_{ij}}{\sum\limits_{i, j = 1}^N \omega_{ij} d_{ij} \delta_{ij}}.
    $$
\end{statement}

\begin{proof}
    The task is to minimize the following expression with respect to $\alpha$:
    $$
    F(\alpha R; \mathbb{X}_N) = \sqrt{\frac{\sum\limits_{i, j = 1}^N \omega_{ij} (\alpha d_{ij} - \delta_{ij})^2}{\sum\limits_{i, j = 1}^N \omega_{ij} \alpha^2 d_{ij}^2}}.
    $$
    To find the optimal $\alpha$, we square $F(\alpha R; \mathbb{X}_N)$ and solve for $\alpha$ by setting the derivative of $F^2$ with respect to $\alpha$ to zero:
    $$
    \frac{d}{d\alpha} \left( F^2 \right) = 0.
    $$
    By solving this equation and verifying the sign of the derivative in the neighborhood of the solution, we confirm the statement.
\end{proof}

This dimension reduction framework provides a principled approach to designing rarity meters that are both interpretable and aligned with market dynamics. By leveraging the inherent structure of NFT trade data, it bridges the gap between theoretical rarity assessment and practical market applications.

\section{Optimal Rarity Meter Design via Unidimensional Scaling}  
\label{sec:OptimalRarityMeterDesign}  

The optimal design for a rarity meter under $F(R; \mathbb{X}_N)$ essentially represents NWUDS. The problem of finding the best configuration, i.e., minimizing stress, has been approached using various methods, from gradient-based to Monte Carlo techniques \cite{Guttman1968Dec,Brusco2001Jan}. The main difficulty of this problem is the vast number of local extrema of the stress function~\cite{Pliner1996}. Despite this, several methods can be applied under certain conditions~\cite{Hubert2002}. Without loss of generality, we will consider arbitrary vectors, each component of which represents one point, as $\vec{x}$. For example, for rarity meters $\vec{x} = (R(X_1), \dots, R(X_N))^T$. The stress to optimize is denoted $S(\vec{x})$.

Pliner \cite{Pliner1986,Pliner1996} suggested minimizing a smoothed version of the multiextremal stress function. Following his algorithm developed for a special case of weights $w = d^p$, for arbitrary non-negative weights, we can write:  
$$  
S(\vec{x}) = \sum\limits_{i, j = 1}^N w_{ij}(d_{ij} - \delta_{ij})^2,  
$$  
$$  
S_\varepsilon(\vec{x}) = \frac{1}{\varepsilon^{N}} \int\limits_{D(\vec{x}, \varepsilon)} S(\vec{y}) d\vec{y},  
$$  
where $D(\vec{x}, \varepsilon)$ is a cube in $\mathbb{R}^N$ with center at $\vec{x}$ and side  $\varepsilon$. Evaluating this integral yields:  
$$  
S_\varepsilon(\vec{x}) = \sum\limits_{i < j} w_{ij} \left( (x_i - x_j)^2 - 2 d_{ij} g_\varepsilon(x_i - x_j) \right) + c,  
$$  
where $c$ is a constant and:  
$$  
g_\varepsilon(t) = \begin{cases}  
\frac{t^2 (3\varepsilon - |t|)}{3 \varepsilon^2} + \frac{\varepsilon}{3}, & |t| < \varepsilon, \\  
|t|, & |t| \geq \varepsilon.  
\end{cases}  
$$  

As $S_\varepsilon(\vec{x})$ is twice continuously differentiable, we can derive a necessary condition for the minimum by equating its partial derivatives to zero:  
$$  
x_i = \sum\limits_{n = 1}^N w_{in} \left( x_n - d_{in} u_\varepsilon (x_i - x_n) \right) \Big/ \sum\limits_{n = 1}^N w_{in},  
$$  
where $i = 1, \ldots, N$, and:  
$$  
u_\varepsilon(t) = \begin{cases}  
\frac{t}{\varepsilon} \left( 2 - \frac{|t|}{\varepsilon} \right), & |t| < \varepsilon, \\  
\text{sign}(t), & |t| \geq \varepsilon.  
\end{cases}  
$$  

The minimum of $S_\varepsilon(\vec{x})$  can be reached with the following iterative algorithm:  
$$  
x_i^{q+1} = \sum\limits_{j = 1}^N w_{ij} \left( x_j^q - d_{ij} u_\varepsilon (x_i^q - x_j^q) \right) \Big/ \sum\limits_{j = 1}^N w_{ij}.  
$$  

\begin{theorem}
    For any initial approximation $\vec{x}^0$, $\lim\limits_{q \rightarrow \infty} \nabla S_\varepsilon (\vec{x}^q) = 0$, and $S_\varepsilon(\vec{x})$ is monotone non-increasing, i.e., $S_\varepsilon(\vec{x}^{q+1}) \leq S_\varepsilon(\vec{x}^{q})$.  
\end{theorem}
\begin{proof}
    The proof is based on the equivalence of the iterative update rule to the gradient method with a constant step-size $\alpha = 1/(2N)$. Following Polyak's theorem on the convergence of the gradient method with a constant step-size \cite{Polyak1983}, we verify that:
    \begin{itemize}
        \item $S_\varepsilon(\vec{x})$ is lower-bounded, as $S(\vec{x}) > 0$ and $S_\varepsilon(\vec{x}) > 0$.
        \item The gradient $\nabla S_\varepsilon(\vec{x})$ satisfies the Lipschitz condition with $L = 2N$, ensuring $$\| \nabla S_\varepsilon(\vec{x}) - \nabla S_\varepsilon(\vec{y}) \| \leq L \| \vec{x} - \vec{y} \|.$$
        \item The step-size $\alpha = 1/(2N)$ satisfies $2 / L = 1/N > \alpha$.
    \end{itemize}
    Thus, the iterative algorithm converges to a stationary point where $\nabla S_\varepsilon(\vec{x}) = 0$, and $S_\varepsilon(\vec{x})$ is monotone non-increasing.
\end{proof}

\begin{theorem}
    If $\varepsilon > 4d^*$, where $$d^* = \max\limits_{1 \leq i \leq N} \left( \sum\limits_{j=1}^N d_{ij} \Big/ \sum\limits_{j=1}^N w_{ij} \right),$$ then $0$ is the only minimum of $S_\varepsilon(\vec{x})$.  
\end{theorem}
\begin{proof}
    Let $v_\varepsilon(\vec{x})$ denote the right-hand side of the update rule:
    $$v_\varepsilon(\vec{x}) = \left( \sum\limits_{j=1}^N w_{ij} \left( x_j - d_{ij} u_\varepsilon(x_i - x_j) \right) \Big/ \sum\limits_{j=1}^N w_{ij} \right)_{i=1}^N.$$
    The mapping $v_\varepsilon$ is contracting when $\varepsilon > 4d^*$, as $$\| v_\varepsilon(\vec{x}) - v_\varepsilon(\vec{y}) \| \leq L \| \vec{x} - \vec{y} \|$$ with $L < 1$. This implies that $v_\varepsilon$ has a unique fixed point, which is the solution to $\vec{x} = v_\varepsilon(\vec{x})$. Since $\vec{0} = v_\varepsilon(\vec{0})$, $\vec{0}$ is the unique minimum of $S_\varepsilon(\vec{x})$.
\end{proof}

\begin{theorem}
\( S_\varepsilon(\vec{x}) \) approximates \( S(\vec{x}) \), and the global minimum of \( S_\varepsilon(\vec{x}) \) approximates the global minimum of \( S(\vec{x}) \) as \( \varepsilon \rightarrow 0 \), in the following sense:  

\begin{enumerate}  
  \item \( S_\varepsilon(\vec{x}) \rightarrow S(\vec{x}) \) uniformly on any bounded subset of \( \mathbb{R}^N \) as \( \varepsilon \rightarrow 0 \).  
  \item If \( \vec{x}^* \) is the only global minimum of \( S(\vec{x}) \), then there exists \( \varepsilon_0 > 0 \) such that for all \( \varepsilon < \varepsilon_0 \), \( S_\varepsilon(\vec{x}) \) also has a unique global minimum \( \vec{x}^\varepsilon \), and \( \vec{x}^\varepsilon \rightarrow \vec{x}^* \). Otherwise, for any \( \vec{x}^* \) global minimum of \( S(\vec{x}) \), \( \inf\limits_{\vec{x}^\varepsilon \in X^\varepsilon} \| \vec{x}^\varepsilon - \vec{x}^* \| \xrightarrow[\varepsilon \rightarrow 0]{} 0 \), where \( X^\varepsilon \) is the set of all global minima of \( S(\vec{x}) \).  
\end{enumerate}  
\end{theorem}
\begin{proof}  
\textbf{Uniform Convergence of \( S_\varepsilon(\vec{x}) \): }
   For any bounded subset \( B \subset \mathbb{R}^N \), the smoothed function \( S_\varepsilon(\vec{x}) \) is defined as:  
   \[
   S_\varepsilon(\vec{x}) = \frac{1}{\varepsilon^N} \int_{D(\vec{x}, \varepsilon)} S(\vec{y}) \, d\vec{y},
   \]  
   where \( D(\vec{x}, \varepsilon) \) is a cube centered at \( \vec{x} \) with side length \( \varepsilon \). As \( \varepsilon \rightarrow 0 \), the integral converges to \( S(\vec{x}) \) uniformly on \( B \) because \( S(\vec{x}) \) is continuous and bounded on \( B \).  

\textbf{Convergence of Global Minima in Case of Unique Global Minimum:} If \( S(\vec{x}) \) has a unique global minimum \( \vec{x}^* \), then for any sequence \( \varepsilon_k \rightarrow 0 \), the corresponding global minima \( \vec{x}^{\varepsilon_k} \) of \( S_{\varepsilon_k}(\vec{x}) \) converge to \( \vec{x}^* \). This follows from the uniform convergence of \( S_\varepsilon(\vec{x}) \) and the continuity of \( S(\vec{x}) \).  
    
\item \textbf{Convergence of Global Minima in Case of Multiple Global Minima:} If \( S(\vec{x}) \) has multiple global minima, then for any global minimum \( \vec{x}^* \) of \( S(\vec{x}) \), the distance between \( \vec{x}^* \) and the set of global minima \( X^\varepsilon \) of \( S_\varepsilon(\vec{x}) \) approaches zero as \( \varepsilon \rightarrow 0 \). This is a consequence of the uniform convergence of \( S_\varepsilon(\vec{x}) \) and the compactness of the set of global minima.
\end{proof}  

Given this, the whole minimization procedure can be outlined as follows:
\begin{itemize}  
  \item Build a sequence \( \{\varepsilon_i\}_{i=1}^I \), such that \( \varepsilon_1 > \varepsilon_2 > \ldots > \varepsilon_I > 0 \). One may select \( \varepsilon_1 = 2d^* \) and \( \varepsilon_i = \varepsilon_1 (I - i + 1) / I \) for \( i = 2, \ldots, I \), where \( d^* = \max\limits_{1 \leq i \leq I} \left( \sum\limits_{j=1}^I d_{ij} \Big/ \sum\limits_{j=1}^I w_{ij} \right) \). 
  \item Minimize \( S_{\varepsilon_1}(\vec{x}) \) using the iterative algorithm, and take its minimization result as the initial configuration for minimizing \( S_{\varepsilon_2}(\vec{x}) \). Repeat this process for all \( \varepsilon_i \) in the sequence.  
\end{itemize}  

This procedure ensures that the solution converges to a global minimum by gradually refining the configuration from a smoothed version of the stress function to the actual stress function.

\section{DIT Rarity Meter}  
\label{section:DITMethodology}  

The proposed algorithm, Dissimilarities in Trades (DIT), is designed to construct an optimal NFT rarity meter by leveraging pairwise dissimilarities derived from NFT trade data. The methodology builds on the theoretical framework of NWUDS and introduces a systematic approach to minimize the stress function $S(\vec{x})$, ensuring alignment with market dynamics.

To evaluate the performance of the DIT algorithm, we split the dataset of trades into training and test sets \cite{Bishop2006}. The test set is used to assess the generalizability of the rarity scores to previously unseen tokens. The performance is measured using $F$ on the test set, ensuring that the rarity scores remain consistent with market dynamics even for out-of-sample tokens.  

\subsection{DIT Training}  
The DIT training algorithm is outlined in Algorithm \ref{alg:OptimalRarityMeter}.  

\begin{algorithm}[h]  
\caption{DIT Training}  
\label{alg:OptimalRarityMeter}  
\begin{algorithmic}[1]  
\REQUIRE Dissimilarity matrix \( (\delta_{ij})_{i,j=1}^N \), weights matrix \( (w_{ij})_{i,j=1}^N \), and initial configuration \( \vec{x}^0 \).  
\ENSURE Optimal rarity scores \( \vec{R} = (R_1, \dots, R_N)^T \).  
\STATE Initialize \( \varepsilon_1 = 2d^* \), where $$d^* = \max\limits_{1 \leq i \leq N} \left( \sum\limits_{j=1}^N d_{ij} \Big/ \sum\limits_{j=1}^N w_{ij} \right).$$  
\STATE Define a sequence \( \{\varepsilon_q\}_{q=1}^Q \) such that \( \varepsilon_q = \varepsilon_1 (Q - q + 1) / Q \) for \( q = 2, \ldots, Q \).  
\FOR{\( q = 1 \) to \( Q \)}  
  \STATE Minimize \( S_{\varepsilon_q}(\vec{x}) \) using the iterative update:  
  \STATE \( x_n^{q+1} = \sum\limits_{j = 1}^N w_{nj} \left( x_j^q - d_{nj} u_\varepsilon (x_n^q - x_j^q) \right) \Big/ \sum\limits_{j = 1}^N w_{nj} \).  
  \STATE Use the result as the initial configuration for the next \( \varepsilon_{q+1} \).  
\ENDFOR  
\STATE Perform a local minimization of \( S(\vec{x}) \) using the final configuration from the previous step.  
\STATE Shift the solution so that \( \min_{n=1, \dots, N} R_n = 1 \).  
\STATE Choose the sign of the solution to maximize the total rarity-weighted log-price.  
\STATE Return the optimal rarity scores \( \vec{R} \).  
\end{algorithmic}  
\end{algorithm} 

\subsection{DIT Test}
\label{subsec:DITTest}
To address the challenges of overfitting and out-of-sample extension (OoS) \cite{Bengio2003}, we apply a nonparametric regression approach for propagating rarity scores to new tokens \cite{Yanovich2017a}. The smoothed rarity score $R_k(X)$ for an NFT $X$ is computed as:  
$$  
R_k(X) = \frac{\sum_{n=1}^N R(X_n) \cdot K\left(|\frac{\tilde{R}(X_n) - \tilde{R}(X)|}{\varepsilon_k(X; \mathbb{X}_N)}\right)}{\sum_{n=1}^N K\left(\frac{|\tilde{R}(X_n) - \tilde{R}(X)|}{\varepsilon_k(X; \mathbb{X}_N)}\right)},  
$$  
where:  
\begin{itemize}  
  \item $\tilde{R}$ is an interpretable rarity meter based on NFT features/traits (like Rarity.tools or KRAMER),  
  $K(a)$ is the exponential kernel:  
  $K(a) = \exp(-a).$  
  \item $\varepsilon_k(X; \mathbb{X}_N)$ is the distance from $X$ to its $k$-th nearest neighbor in $\mathbb{X}_N$.  
\end{itemize}  

The interpretable rarity meter $\tilde{R}$ and the parameter $k$ are selected via grid search during cross-validation on the training set. This ensures the robustness of the rarity scores and their generalizability to new tokens.   

By integrating dimension reduction, smoothing, and interpretability, DIT provides a principled and practical approach to NFT rarity assessment, bridging the gap between theoretical rarity metrics and real-world market applications.

\section{Numerical Experiments}  
\label{sec:Experiments}  

In this section, we present the results of our numerical experiments using the updated ROAR benchmark dataset, which includes trade data up to November 18, 2024 \cite{Belousov2024}. We evaluate the performance of the proposed DIT algorithm, an non-interpretable rarity meter, in comparison to existing interpretable rarity meters. Additionally, we analyze the impact of key parameters on the performance of the DIT algorithm. The dataset and experimental scripts used in this study are openly accessible on GitHub \cite{fasghq2025}.

\subsection{Dataset Overview}  
The ROAR benchmark dataset comprises $100$ NFT collections in Ethereum blockchain top ranked by trade volume, with $4.1$ millions of total trades. Each collection includes detailed trait and trade information, enabling a comprehensive evaluation of rarity meters.  

\subsection{Performance Measure}  
We adopt the NWUDS-inspired DIT's $F$ as the performance measure, replacing the previously employed weighted correlation $F_{wc}$. This modification is motivated by its computational efficiency and its natural alignment with the dimension reduction framework. The measure assesses the similarity between pairwise rarities and deal prices, offering a robust basis for comparative analysis. All rarity meters were scale-adjusted using $\alpha$ from Statement~\ref{statement:scale} to optimize their performance.

\subsection{Comparison with State-of-the-Art}  
We compare the performance of DIT against five state-of-the-art rarity meters: Rarity.tools, KRAMER, OpenRarity, NFTGO, and ROAR. Performance profiles are used to visualize the results across the benchmark dataset.  

\begin{figure}[h!]  
    \includegraphics[width=0.99\linewidth]{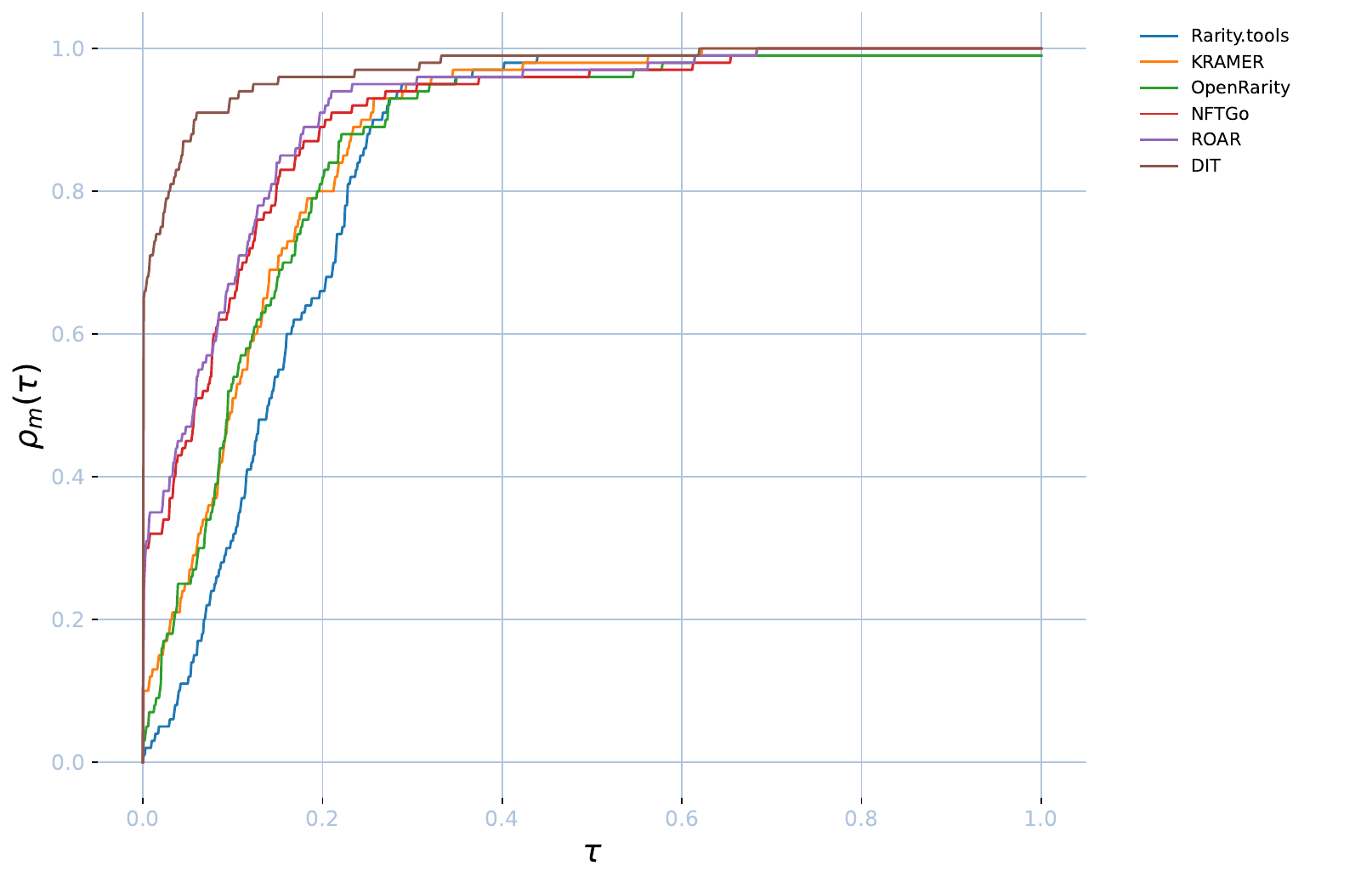}  
    \centering  
    \caption{Performance profiles of rarity meters on the updated ROAR benchmark dataset.}  
    \label{fig:PerformanceProfiles}  
\end{figure}  

Figure~\ref{fig:PerformanceProfiles} shows that DIT consistently outperforms the other methods, achieving the best performance in over $65\%$ of the collections while maintaining a competitive edge in the remaining cases, albeit at the cost of interpretability. ROAR and NFTGo follow in second and third place, respectively. 

Notably, in the most of the cases NFTGo was the best choice for interpretable rarity meter $\tilde{R}$ during DIT test (see subsection \ref{subsec:DITTest}). The optimal $k$ was different from collection to collection, so we will ellaborate more on it in the next subsection.

\subsection{Impact of Neighborhood Parameter}
We investigate the effect of the neighborhood parameter $k$ on the performance of DIT. This parameter controls the number of neighbors used in the nonparametric regression for out-of-sample extension.  

\begin{figure}[h!]  
    \includegraphics[width=0.95\linewidth]{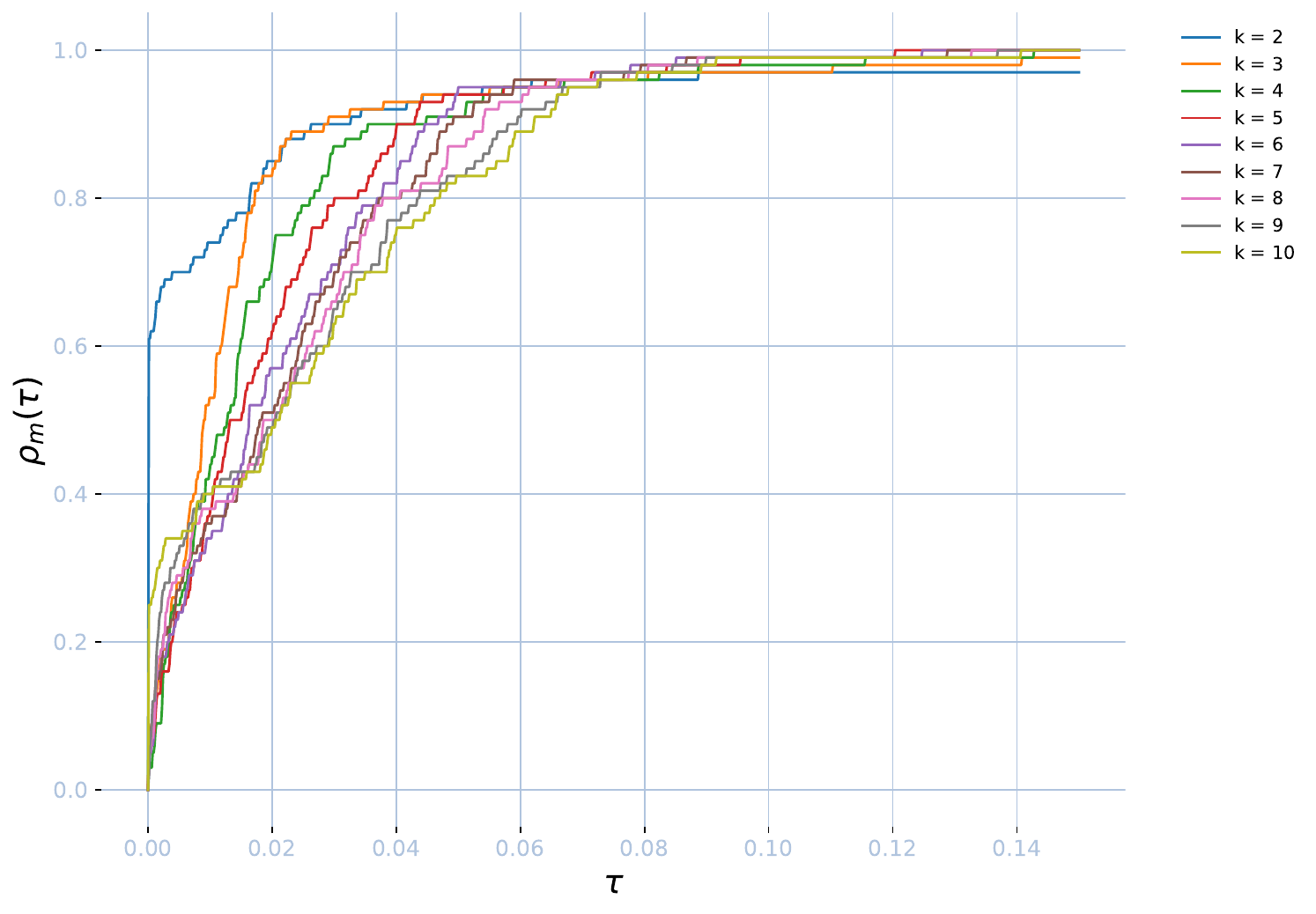}   
    \centering  
    \caption{Impact of the neighborhood parameter $k$ on the performance of DIT.}  
    \label{fig:kParameterImpact}  
\end{figure}  

Figure~\ref{fig:kParameterImpact} demonstrates that the optimal $k$ was $2$ in $60\%$, while the others from $1$ to $25$ had fractions from $1\%$ to $4\%$.  

\section{Discussion and Implications}  
\label{sec:Discussion}  

The proposed framework for NFT rarity meters offers several key advantages for blockchain-based information systems and business process management. First, it provides a rigorous mathematical foundation for assessing rarity, grounded in the principles of dimension reduction. This approach enhances the reliability and transparency of NFT valuation, which is critical for building trust in decentralized systems. Second, the framework accounts for the inherent volatility of NFT markets by focusing on pairs of deals that are close in time, thereby mitigating the impact of price fluctuations. This is particularly relevant for blockchain-based systems, where real-time data integrity and accuracy are paramount. Third, the iterative algorithm ensures computational efficiency and convergence of the DIT rarity meter to a globally optimal solution, making it scalable for large-scale applications.

The implications of this work extend beyond NFTs. The framework can be adapted to other domains where rarity or dissimilarity measures are critical, such as digital art, collectibles, and financial assets. Future research could explore the development of real-time rarity assessment tools and applications in anti-money laundering, leveraging blockchain’s immutability and transparency. Additionally, the framework could be integrated into blockchain-based business process management systems to enhance decision-making and optimize resource allocation.

The experimental results highlight the superiority of DIT in aligning NFT rarity scores with market dynamics. The algorithm’s ability to leverage pairwise dissimilarities in trades, combined with efficient out-of-sample extension, makes it a practical and effective tool for NFT rarity assessment. By adopting unidimensional scaling as the performance measure, we ensure a computationally efficient and theoretically motivated evaluation framework.

These findings underscore the importance of integrating market data and advanced optimization techniques in designing NFT rarity meters, paving the way for future research in blockchain-based information systems and business process management.

\section{Conclusion}  
\label{sec:Conclusion}  

In this paper, we have introduced a novel dimension reduction perspective for designing optimal NFT rarity meters, grounded in the principles of non-metric weighted unidimensional scaling. Our framework bridges the gap between theoretical rarity assessment and practical market applications, ensuring that the resulting rarity scores are robustly aligned with real-world market dynamics. The proposed DIT algorithm demonstrates superior performance compared to existing methods, albeit at the cost of reduced interpretability. Despite this limitation, the DIT algorithm can still be effectively integrated into automated tools for NFT traders and utilized by law enforcement to detect market manipulation such as pump-and-dump schemes and wash trading \cite{Morgia2023}.

Our work underscores the critical importance of incorporating market data into the design of rarity meters and establishes a solid foundation for future research in blockchain-based information systems. As the NFT market continues to evolve, the ability to accurately assess and quantify rarity will remain a pivotal challenge. The methods presented here represent a significant step forward, with potential applications in anti-money laundering (AML), real-time rarity assessment, and blockchain-based business process management (BPM).

Future research directions could focus on extending this framework by integrating additional market features, developing real-time assessment tools, and addressing scalability and interoperability challenges inherent in blockchain-based systems. By advancing the understanding of NFT rarity and its applications, this work contributes to the broader objective of building trustworthy, efficient, and transparent blockchain-based information systems.

\section*{Acknowledgment}

The authors are thankful to Maksim Shuklin, who participated in the early research stage in 2023/2024 academic year.

We acknowledge the use of ChatGPT and DeepSeek in enhancing the readability and clarity of this manuscript. These tools were employed to assist in refining language and improving the overall presentation of the content. However, the authors retain full responsibility for the integrity, accuracy, and intellectual contributions of the research at all stages.

\bibliographystyle{elsarticle-num}
\bibliography{refs_2025}

\begin{thebibliography}{10}
\expandafter\ifx\csname url\endcsname\relax
  \def\url#1{\texttt{#1}}\fi
\expandafter\ifx\csname urlprefix\endcsname\relax\def\urlprefix{URL }\fi
\expandafter\ifx\csname href\endcsname\relax
  \def\href#1#2{#2} \def\path#1{#1}\fi

\bibitem{Oliveira2018}
L.~Oliveira, I.~Bauer, L.~Zavolokina, G.~Schwabe, {To token or not to token: Tools for understanding blockchain tokens}, in: International Conference on Information Systems 2018, ICIS 2018, 2018, pp. 1--17.

\bibitem{Angelo2020}
M.~di~Angelo, G.~Salzer, \href{https://ieeexplore.ieee.org/document/9126009/}{{Tokens, Types, and Standards: Identification and Utilization in Ethereum}}, in: 2020 IEEE International Conference on Decentralized Applications and Infrastructures (DAPPS), IEEE, 2020, pp. 1--10.
\newblock \href {https://doi.org/10.1109/DAPPS49028.2020.00001} {\path{doi:10.1109/DAPPS49028.2020.00001}}.
\newline\urlprefix\url{https://ieeexplore.ieee.org/document/9126009/}

\bibitem{Wang2021}
N.~Wang, S.~Chi-Kin~Chau, Y.~Zhou, \href{https://doi.org/10.1145/3447555.3464869}{{Privacy-Preserving Energy Storage Sharing with Blockchain}}, in: Proceedings of the Twelfth ACM International Conference on Future Energy Systems, ACM, New York, NY, USA, 2021.
\newline\urlprefix\url{https://doi.org/10.1145/3447555.3464869}

\bibitem{Wood2016}
G.~Wood, \href{https://github.com/w3f/polkadot-white-paper/raw/master/PolkaDotPaper.pdf}{{Polkadot: Vision for a Heterogeneous Multi-Chain Framework}}, Whitepaper (2017) 1--21\href {https://doi.org/10.1021/acs.jpcc.6b00269} {\path{doi:10.1021/acs.jpcc.6b00269}}.
\newline\urlprefix\url{https://github.com/w3f/polkadot-white-paper/raw/master/PolkaDotPaper.pdf}

\bibitem{Nakamoto2008}
S.~Nakamoto, \href{https://bitcoin.org/bitcoin.pdf}{{Bitcoin: A Peer-to-Peer Electronic Cash System}}, www.bitcoin.org (2008) 1--9.
\newline\urlprefix\url{https://bitcoin.org/bitcoin.pdf}

\bibitem{Yakovenko2018}
A.~Yakovenko, \href{https://solana.com/solana-whitepaper.pdf}{{Solana: A new architecture for a high performance blockchain}} (2018).
\newline\urlprefix\url{https://solana.com/solana-whitepaper.pdf}

\bibitem{Rarity.tools2021}
{Rarity.tools}, \href{https://raritytools.medium.com/ranking-rarity-understanding-rarity-calculation-methods-86ceaeb9b98c}{{Ranking Rarity: Understanding Rarity Calculation Methods}} (2021).
\newline\urlprefix\url{https://raritytools.medium.com/ranking-rarity-understanding-rarity-calculation-methods-86ceaeb9b98c}

\bibitem{Krasnoselskii2023}
M.~Krasnoselskii, Y.~Madhwal, Y.~Yanovich, \href{https://ieeexplore.ieee.org/document/10014994/}{{KRAMER: Interpretable Rarity Meter for Crypto Collectibles}}, IEEE Access 11 (2023) 4283--4290.
\newblock \href {https://doi.org/10.1109/ACCESS.2023.3236080} {\path{doi:10.1109/ACCESS.2023.3236080}}.
\newline\urlprefix\url{https://ieeexplore.ieee.org/document/10014994/}

\bibitem{Belousov2024}
D.~Belousov, M.~Shuklin, A.~Stepin, Y.~Yanovich, {ROAR: A Benchmark for NFT Rarity Meters}, in: 2024 IEEE International Conference on Blockchain and Cryptocurrency (ICBC), IEEE, 2024.

\bibitem{PauleVianez2022}
J.~Paule-Vianez, A.~Alc{\'{a}}zar-Blanco, J.~L. Coca-P{\'{e}}rez, {Effect of Economic Policy Uncertainty on the investment in numismatic assets: Evidence for the Walking Liberty Half Dollar}, Finance Research Letters 46 (2022) 102412.
\newblock \href {https://doi.org/10.1016/J.FRL.2021.102412} {\path{doi:10.1016/J.FRL.2021.102412}}.

\bibitem{Chailan2018}
C.~Chailan, {Art as a means to recreate luxury brands' rarity and value}, Journal of Business Research 85 (2018) 414--423.
\newblock \href {https://doi.org/10.1016/J.JBUSRES.2017.10.019} {\path{doi:10.1016/J.JBUSRES.2017.10.019}}.

\bibitem{Koford1998}
K.~Koford, A.~E. Tschoegl, {The market value of rarity}, Journal of Economic Behavior {\&} Organization 34~(3) (1998) 445--457.
\newblock \href {https://doi.org/10.1016/S0167-2681(97)00084-X} {\path{doi:10.1016/S0167-2681(97)00084-X}}.

\bibitem{Hastie2009}
T.~Hastie, R.~Tibshirani, J.~Friedman, {The Elements of Statistical Learning}, Springer Series in Statistics, Springer New York, New York, NY, 2009.
\newblock \href {https://doi.org/10.1007/978-0-387-84858-7} {\path{doi:10.1007/978-0-387-84858-7}}.

\bibitem{Hughes2022}
J.~E. Hughes, \href{https://onlinelibrary.wiley.com/doi/10.1111/joie.12262}{{Demand for Rarity: Evidence from a Collectible Good}}, The Journal of Industrial Economics 70~(1) (2022) 147--167.
\newblock \href {https://doi.org/10.1111/JOIE.12262} {\path{doi:10.1111/JOIE.12262}}.
\newline\urlprefix\url{https://onlinelibrary.wiley.com/doi/10.1111/joie.12262}

\bibitem{SlashdotMedia2022}
{Slashdot Media}, \href{https://sourceforge.net/software/product/rarity.tools/alternatives}{{Best rarity.tools Alternatives {\&} Competitors}} (2022).
\newline\urlprefix\url{https://sourceforge.net/software/product/rarity.tools/alternatives}

\bibitem{Krasnoselskii2022}
M.~Krasnoselskii, Y.~Madhwal, Y.~Yanovich, \href{https://ieeexplore.ieee.org/document/9805542/}{{KRAMER: Kanaria NFT Collection Rarity Meter}}, in: 2022 IEEE International Conference on Blockchain and Cryptocurrency (ICBC), IEEE, 2022, pp. 1--2.
\newblock \href {https://doi.org/10.1109/ICBC54727.2022.9805542} {\path{doi:10.1109/ICBC54727.2022.9805542}}.
\newline\urlprefix\url{https://ieeexplore.ieee.org/document/9805542/}

\bibitem{White2022}
B.~White, A.~Mahanti, K.~Passi, {Characterizing the OpenSea NFT Marketplace}, in: Companion Proceedings of the Web Conference 2022, ACM, New York, NY, USA, 2022, pp. 488--496.
\newblock \href {https://doi.org/10.1145/3487553.3524629} {\path{doi:10.1145/3487553.3524629}}.

\bibitem{OpenRarity2022}
{OpenRarity}, \href{https://mirror.xyz/openrarity.eth/-R8ZA5KCMgqtsueySlruAhB77YBX6fSnS_dT-8clZPQ}{{Introducing OpenRarity}} (2022).
\newline\urlprefix\url{https://mirror.xyz/openrarity.eth/-R8ZA5KCMgqtsueySlruAhB77YBX6fSnS_dT-8clZPQ}

\bibitem{Yan2023}
F.~Yan, X.~Wang, K.~Mao, W.~Zhang, W.~Chen, {NFTVis: Visual Analysis of NFT Performance}, in: 2023 IEEE 16th Pacific Visualization Symposium (PacificVis), IEEE, 2023, pp. 82--91.
\newblock \href {https://doi.org/10.1109/PacificVis56936.2023.00016} {\path{doi:10.1109/PacificVis56936.2023.00016}}.

\bibitem{NFTGo2023}
{NFTGo}, \href{https://docs.nftgo.io/}{{NFTGo: The World's Leading NFT Data Intelligence Provider}}.
\newline\urlprefix\url{https://docs.nftgo.io/}

\bibitem{Lee2024}
H.~Lee, G.-C. Lee, H.-Y. Koo, {Exploring the relationship between rarity and price of profile picture NFT: A formal concept analysis on the BAYC NFT collection}, Blockchain: Research and Applications (2024).
\newblock \href {https://doi.org/10.1016/j.bcra.2024.100191} {\path{doi:10.1016/j.bcra.2024.100191}}.

\bibitem{Nadini2021}
M.~Nadini, L.~Alessandretti, F.~Di~Giacinto, M.~Martino, L.~M. Aiello, A.~Baronchelli, {Mapping the NFT revolution: market trends, trade networks, and visual features}, Scientific Reports 11~(1) (2021) 20902.
\newblock \href {https://doi.org/10.1038/s41598-021-00053-8} {\path{doi:10.1038/s41598-021-00053-8}}.

\bibitem{Mekacher2022}
A.~Mekacher, A.~Bracci, M.~Nadini, M.~Martino, L.~Alessandretti, L.~M. Aiello, A.~Baronchelli, \href{https://www.nature.com/articles/s41598-022-17922-5}{{Heterogeneous rarity patterns drive price dynamics in NFT collections}}, Scientific Reports 12~(1) (2022) 13890.
\newblock \href {https://doi.org/10.1038/s41598-022-17922-5} {\path{doi:10.1038/s41598-022-17922-5}}.
\newline\urlprefix\url{https://www.nature.com/articles/s41598-022-17922-5}

\bibitem{Schaar2022}
L.~Schaar, S.~Kampakis, {Non-fungible Tokens as an Alternative Investment: Evidence from CryptoPunks}, The Journal of The British Blockchain Association 5~(1) (2022) 1--12.
\newblock \href {https://doi.org/10.31585/jbba-5-1-(2)2022} {\path{doi:10.31585/jbba-5-1-(2)2022}}.

\bibitem{Aydogdu2024}
D.~Aydoğdu, N.~Aydin, \href{https://ieeexplore.ieee.org/document/10786976/}{Development of a hybrid recommendation system for nfts using deep learning techniques}, IEEE Access (2024).
\newblock \href {https://doi.org/10.1109/ACCESS.2024.3514512} {\path{doi:10.1109/ACCESS.2024.3514512}}.
\newline\urlprefix\url{https://ieeexplore.ieee.org/document/10786976/}

\bibitem{Pala2024}
M.~Pala, E.~Sefer, Nft price and sales characteristics prediction by transfer learning of visual attributes, The Journal of Finance and Data Science 10 (2024) 100148.
\newblock \href {https://doi.org/10.1016/j.jfds.2024.100148} {\path{doi:10.1016/j.jfds.2024.100148}}.

\bibitem{Pearson1901}
K.~Pearson, {On lines and planes of closest fit to systems of points in space}, Philosophical Magazine 2~(11) (1901) 559–572.
\newblock \href {https://doi.org/10.1080/14786440109462720} {\path{doi:10.1080/14786440109462720}}.

\bibitem{Cox2008}
M.~A.~A. Cox, T.~F. Cox, \href{http://link.springer.com/10.1007/978-3-540-33037-0_14}{{Multidimensional Scaling}}, in: Handbook of Data Visualization, Springer Berlin Heidelberg, Berlin, Heidelberg, 2008, pp. 315--347.
\newblock \href {https://doi.org/10.1007/978-3-540-33037-0{\_}14} {\path{doi:10.1007/978-3-540-33037-0{\_}14}}.
\newline\urlprefix\url{http://link.springer.com/10.1007/978-3-540-33037-0_14}

\bibitem{Leeuw2009}
J.~d. Leeuw, P.~Mair, {Multidimensional Scaling Using Majorization: SMACOF in R}, Journal of Statistical Software 31~(3) (2009).
\newblock \href {https://doi.org/10.18637/jss.v031.i03} {\path{doi:10.18637/jss.v031.i03}}.

\bibitem{Seung2000}
H.~S. Seung, \href{http://www.sciencemag.org/cgi/doi/10.1126/science.290.5500.2268}{{The Manifold Ways of Perception}}, Science 290~(5500) (2000) 2268--2269.
\newblock \href {https://doi.org/10.1126/science.290.5500.2268} {\path{doi:10.1126/science.290.5500.2268}}.
\newline\urlprefix\url{http://www.sciencemag.org/cgi/doi/10.1126/science.290.5500.2268}

\bibitem{Levina2004}
E.~Levina, P.~J. Bickel, \href{https://www.stat.berkeley.edu/~bickel/mldim.pdf}{{Maximum Likelihood Estimation of Intrinsic Dimension}}, in: Advances in Neural Information Processing Systems, MIT Press, 2005, pp. 777--784.
\newline\urlprefix\url{https://www.stat.berkeley.edu/~bickel/mldim.pdf}

\bibitem{Gomtsyan2019}
M.~Gomtsyan, N.~Mokrov, M.~Panov, Y.~Yanovich, \href{http://arxiv.org/abs/1904.06151}{{Geometry-Aware Maximum Likelihood Estimation of Intrinsic Dimension}}, Proceedings of Machine Learning Research 101 (2019) 1126--1141.
\newline\urlprefix\url{http://arxiv.org/abs/1904.06151}

\bibitem{Belkin2003}
M.~Belkin, P.~Niyogi, \href{https://direct.mit.edu/neco/article/15/6/1373-1396/6730}{{Laplacian Eigenmaps for Dimensionality Reduction and Data Representation}}, Neural Computation 15~(6) (2003) 1373--1396.
\newblock \href {https://doi.org/10.1162/089976603321780317} {\path{doi:10.1162/089976603321780317}}.
\newline\urlprefix\url{https://direct.mit.edu/neco/article/15/6/1373-1396/6730}

\bibitem{Huo2008}
X.~Huo, X.~S. Ni, A.~K. Smith, \href{http://www.worldscientific.com/doi/abs/10.1142/9789812779861_0015}{{A Survey of Manifold-Based Learning Methods}}, in: Recent Advances in Data Mining of Enterprise Data: Algorithms and Applications, WORLD SCIENTIFIC, 2008, pp. 691--745.
\newblock \href {https://doi.org/10.1142/9789812779861{\_}0015} {\path{doi:10.1142/9789812779861{\_}0015}}.
\newline\urlprefix\url{http://www.worldscientific.com/doi/abs/10.1142/9789812779861_0015}

\bibitem{VectorDiffusionMaps}
A.~Singer, H.-T. Wu, \href{http://doi.wiley.com/10.1002/cpa.21395}{{Vector diffusion maps and the connection Laplacian}}, Communications on Pure and Applied Mathematics 65~(8) (2012) 1067--1144.
\newblock \href {https://doi.org/10.1002/cpa.21395} {\path{doi:10.1002/cpa.21395}}.
\newline\urlprefix\url{http://doi.wiley.com/10.1002/cpa.21395}

\bibitem{Bernstein2015b}
A.~Bernstein, A.~Kuleshov, Y.~Yanovich, \href{http://ieeexplore.ieee.org/document/7344815/}{{Information preserving and locally isometric{\&}amp;conformal embedding via Tangent Manifold Learning}}, in: Proceedings of the 2015 IEEE International Conference on Data Science and Advanced Analytics (DSAA), IEEE, 2015, pp. 1--9.
\newblock \href {https://doi.org/10.1109/DSAA.2015.7344815} {\path{doi:10.1109/DSAA.2015.7344815}}.
\newline\urlprefix\url{http://ieeexplore.ieee.org/document/7344815/}

\bibitem{Baldi2012}
P.~Baldi, {Autoencoders, Unsupervised Learning, and Deep Architectures}, ICML Unsupervised and Transfer Learning (2012).
\newblock \href {https://doi.org/10.1561/2200000006} {\path{doi:10.1561/2200000006}}.

\bibitem{Bank2023}
D.~Bank, N.~Koenigstein, R.~Giryes, Autoencoders, Springer International Publishing, 2023, pp. 353--374.
\newblock \href {https://doi.org/10.1007/978-3-031-24628-9_16} {\path{doi:10.1007/978-3-031-24628-9_16}}.

\bibitem{Ross2022}
B.~L. Ross, G.~Loaiza-Ganem, A.~L. Caterini, J.~C. Cresswell, {Neural Implicit Manifold Learning for Topology-Aware Density Estimation} (6 2022).

\bibitem{Pliner1984}
V.~M. Pliner, A class of metric scaling models, Avtomat. i Telemekh. (1984) 122--128.

\bibitem{Pliner1986}
V.~M. Pliner, The problem of multidimensional metric scaling, Avtomat. i Telemekh. (1986) 140--148.

\bibitem{Pliner1996}
V.~Pliner, {Metric unidimensional scaling and global optimization}, J. Classification 13~(1) (1996) 3--18.
\newblock \href {https://doi.org/10.1007/BF01202579} {\path{doi:10.1007/BF01202579}}.

\bibitem{Ge2005}
Y.~Ge, Y.~Leung, J.~Ma, {Unidimensional scaling classifier and its application to remotely sensed data}, in: Proceedings. 2005 IEEE International Geoscience and Remote Sensing Symposium, 2005. IGARSS '05., IEEE, 2005, pp. 3841--3844.
\newblock \href {https://doi.org/10.1109/IGARSS.2005.1525747} {\path{doi:10.1109/IGARSS.2005.1525747}}.

\bibitem{Hubert2002}
L.~J. Hubert, P.~Arabie, J.~J. Meulman, {Linear Unidimensional Scaling in the L 2 -Norm: Basic Optimization Methods Using MATLAB}, Journal of Classification 19~(2) (2002) 303--328.
\newblock \href {https://doi.org/10.1007/s00357-001-0047-5} {\path{doi:10.1007/s00357-001-0047-5}}.

\bibitem{Guttman1968Dec}
L.~Guttman, {A general nonmetric technique for finding the smallest coordinate space for a configuration of points}, Psychometrika 33~(4) (1968) 469--506.
\newblock \href {https://doi.org/10.1007/BF02290164} {\path{doi:10.1007/BF02290164}}.

\bibitem{Brusco2001Jan}
M.~J. Brusco, {A Simulated Annealing Heuristic for Unidimensional and Multidimensional (City-Block) Scaling of Symmetric Proximity Matrices}, J. Classification 18~(1) (2001) 3--33.
\newblock \href {https://doi.org/10.1007/s00357-0003-4} {\path{doi:10.1007/s00357-0003-4}}.

\bibitem{Polyak1983}
B.~Polyak, Introduction to Optimization, Nauka, 1983.

\bibitem{Bishop2006}
C.~M. Bishop, {Pattern Recognition and Machine Learning}, Springer-Verlag, New York, 2006.

\bibitem{Bengio2003}
Y.~Bengio, J.-F. Paiement, P.~Vincent, {Out-of-Sample Extensions for LLE, Isomap, MDS, Eigenmaps, and Spectral Clustering}, In Advances in Neural Information Processing Systems (2003) 177--184\href {https://doi.org/10.1.1.5.1709} {\path{doi:10.1.1.5.1709}}.

\bibitem{Yanovich2017a}
Y.~Yanovich, {Asymptotic Properties of Nonparametric Estimation on Manifold}, JMLR Workshop and Conference Proceedings 60 (2017) 18--38.

\bibitem{fasghq2025}
fasghq, maximshuklin, \href{https://github.com/fasghq/DIT-benchmark/tree/main}{Dit-benchmark: A benchmark for nft rarity meters} (2025).
\newline\urlprefix\url{https://github.com/fasghq/DIT-benchmark/tree/main}

\bibitem{Morgia2023}
M.~L. Morgia, A.~Mei, A.~M. Mongardini, E.~N. Nemmi, {A Game of NFTs: Characterizing NFT Wash Trading in the Ethereum Blockchain}, in: 2023 IEEE 43rd International Conference on Distributed Computing Systems (ICDCS), IEEE, 2023, pp. 13--24.
\newblock \href {https://doi.org/10.1109/ICDCS57875.2023.00018} {\path{doi:10.1109/ICDCS57875.2023.00018}}.

\end{thebibliography}





\end{document}